\documentclass[12pt]{article}
\usepackage[utf8]{inputenc}

\usepackage{amsmath}
\usepackage[margin=1.5in]{geometry}
\usepackage{amsthm}
\usepackage{parskip}
\usepackage{titling}
\usepackage{titlesec}
\usepackage{authblk}
\usepackage{xcolor}
\definecolor{ourblue}{HTML}{0571a3}
\usepackage[colorlinks=true, allcolors=ourblue]{hyperref}
\usepackage{setspace}
\usepackage{caption}

\usepackage{mathptmx}

\newenvironment{journalabstract}{\section*{Abstract}}
{}

\DeclareCaptionLabelFormat{mid}{#1 #2 $\mathbf{\vert}$~}

\newcommand{\SupplementNameFull}{Supplementary Information}

\newcommand{\papertitle}{Optimal, near-optimal, and robust epidemic control}
\title{\papertitle}
\date{}

\author[1,3*]{Dylan H.~Morris}
\author[1*]{Fernando W.~Rossine}
\author[2]{Joshua B.~Plotkin}
\author[1]{Simon A.~Levin}
\affil[1]{Department of Ecology \& Evolutionary Biology, Princeton University, 106A Guyot Hall, Princeton, NJ 08544, USA}
\affil[2]{Department of Biology \& Department of Mathematics, The University of Pennsylvania, 433 S University Ave, Philadelphia, PA 19104, USA}
\affil[3]{Present address: Department of Ecology \& Evolutionary Biology, University of California Los Angeles, Terasaki Life Sciences Building, 610 Charles E. Young Dr South, Los Angeles, CA  90095, USA}
\affil[*]{these authors contributed equally; correspondence to \href{mailto:dylan@dylanhmorris.com}{dylan@dylanhmorris.com},  \href{mailto:frossine@princeton.edu}{frossine@princeton.edu}.}

\usepackage[backend=biber, style=nature]{biblatex}
\usepackage{graphicx}
\usepackage{amsfonts}

\usepackage{physics}
\usepackage{appendix}

\usepackage{macros/optimal-theorem-environments}
\usepackage{macros/optimal-macros} 
\newcommand{\autocaption}[2]{\caption[#1]{\textbf{#1} #2}}
\newcommand{\snref}[1]{Supplementary Note \ref{#1}}

\begin{document}
\maketitle

\captionsetup*[figure]{labelfont={bf}, 
labelformat = mid,
labelsep = none,
name = {Fig.}}
\captionsetup*[table]{labelfont={bf}, 
labelformat = mid,
labelsep = none,
name = {Table}}

\begin{refsection}
\begin{journalabstract}
In the absence of drugs and vaccines, policymakers use non-pharmaceutical interventions such as social distancing to decrease rates of disease-causing contact, with the aim of reducing or delaying the epidemic peak. These measures carry social and economic costs, so societies may be unable to maintain them for more than a short period of time. Intervention policy design often relies on numerical simulations of epidemic models, but comparing policies and assessing their robustness demands clear principles that apply across strategies. Here we derive the theoretically optimal strategy for using a time-limited intervention to reduce the peak prevalence of a novel disease in the classic Susceptible-Infectious-Recovered epidemic model. We show that broad classes of easier-to-implement strategies can perform nearly as well as the theoretically optimal strategy. But neither the optimal strategy nor any of these near-optimal strategies is robust to implementation error: small errors in timing the intervention produce large increases in peak prevalence. Our results reveal fundamental principles of non-pharmaceutical disease control and expose their potential fragility. For robust control, an intervention must be strong, early, and ideally sustained.
\end{journalabstract}

\section*{Introduction}
New human pathogens routinely emerge via zoonotic spillover from other species. Examples include ebolaviruses \autocite{leroy2005fruit}, influenza viruses \autocite{taubenberger2010influenza}, and most recently the sarbecoronavirus SARS-CoV-2 \autocite{zhou2020pneumonia}. Many of these emerging pathogens are antigenically novel: the human population possesses little or no preexisting immunity \autocite{lipsitch2020cross}. This not only can increase disease severity \autocite{dan2021immunological} but also leads to explosive epidemic spread \autocite{lipsitch2020cross}. If left unchecked, that explosive spread can result in a large proportion of the population becoming synchronously infected, as occurred during the COVID-19 global pandemic. Policymakers initially have limited tools for controlling a novel pathogen epidemic; it can take months to develop drugs, and years to develop and distribute vaccines \autocite{graham2018novel}. If disease symptoms are severe, healthcare systems may be strained to the breaking point as the epidemic approaches its peak \autocite{anderson2020will}. 

In the absence of drugs and vaccines, mitigation efforts to reduce or delay the peak (flattening the curve \autocite{branswell2020flattening,anderson2020will}) rely on non-pharmaceutical interventions \autocite{who2006nonpharmaceutical} such as social distancing \autocite{kissler2020social} that decrease rates of disease-transmitting contact. These measures carry social and economic costs, and so societies may be unable to maintain them for more than a short period of time. 

Policy design for allocating non-pharmaceutical resources during the COVID-19 pandemic relied heavily on numerical simulations of epidemic models \autocite{kissler2020social,ferguson2020impact}, but it is difficult to compare predictions or assess robustness without broad principles that apply across strategies. Since models are not reality, robust model-based policy requires not only generating a desired outcome but also understanding what elements of the model are producing it. This typically requires analytical and theoretical understanding.

Relatively little is known about globally optimal strategies for epidemic control, regardless of principal aim \autocite{feng2007final,hollingsworth2011mitigation} or intervention duration. One result establishes that time-limited interventions to reduce peak prevalence should start earlier than time-limited interventions to reduce the final size \autocite{dilauro2020timing}. A number of studies of COVID-19 have used optimal control theory---an approach that relies on numerical optimization to study continuous error-correction \autocite{perkins2020optimal,dilauro2020covidfeedback}.

But the optimal time-limited strategy to reduce peak prevalence is not known. Without an analytical understanding of epidemic peak reduction, policy design based on numerical simulation may fail in unexpected ways; policies may be inefficient, non-robust, or both.

The early months of the COVID-19 pandemic demonstrated how difficult real-time epidemiological modeling, inference, and response can be. Large numbers of asymptomatic and mildly symptomatic cases \autocite{li2020substantial}, as well as difficulties with testing, particularly in the United States \autocite{shear2020lostmonth}, left policymakers with substantial uncertainty regarding the virus's epidemiological parameters and the case numbers in many locations. Countries including the United States and United Kingdom waited to intervene until transmission was widespread; retrospective analyses later claimed that even slightly earlier intervention would have saved many lives \autocite{pei2020differential,knock20212020}.

Epidemiological uncertainties mean that no policymaker can intervene at precisely the optimal time. To understand how this limitation can hinder policymaking, we assess the robustness of optimized interventions to timing error: what is the cost of intervening too early or too late?

We derive the theoretically optimal strategy for using a time-limited intervention to reduce the peak prevalence of a novel disease in the classic Susceptible-Infectious-Recovered (SIR) epidemic model \autocite{kermack1927contribution,weiss2013sir}. We show that broad classes of strategies that are easier to implement can perform nearly as well as this theoretically optimal strategy. But we show that neither the theoretically optimal strategy nor any of these near-optimal strategies is robust to implementation error: small errors in timing of the intervention produce large increases in peak prevalence. To prevent disastrous outcomes due to imperfect implementation, a strong, early, and ideally sustained non-pharmaceutical response is required.

\bigskip
\section*{Results}
\subsection*{Epidemic model and interventions}
We consider the standard susceptible-infectious-recovered (SIR) epidemic model \autocite{kermack1927contribution}, which describes the fractions of susceptible $S(t)$, infectious $I(t)$, and recovered $R(t)$ individuals in the population at time $t$ \autocite{weiss2013sir}. New infections occur proportional to $S(t)I(t)$ at a rate $\beta$, and infectious individuals recover at a rate $\gamma$. The model has a basic reproduction number $\Ro = \frac{\beta}{\gamma}$ and an effective reproduction number $\Reff = \frac{\beta}{\gamma} S(t)$. We denote the peak prevalence by \Imax{}.

We consider interventions that reduce the effective rate of disease-transmitting contacts, $\betat$ below its value in the absence of intervention $\beta$, which we treat as fixed. In the context of COVID-19, this includes non-pharmaceutical interventions such as social distancing \autocite{kissler2020social}, as well as some pharmaceutical ones, such as antivirals that might reduce virus shedding.

We permit interventions to operate on $\beta$ only for a limited duration of time, $\tau$. We impose this constraint in light of political, social, and economic impediments to maintaining aggressive intervention indefinitely. That is, we treat costs of intervention implicitly, subsuming them in $\tau$. We describe such an intervention by defining a transmission reduction function $b(t)$ such that:

\begin{equation}
\label{eqn:intervention-sir}
    \begin{aligned}
        \dv{S}{t} &= -b(t) *\beta SI\\
        \dv{I}{t} &= b(t) *\beta SI - \gamma I\\
        \dv{R}{t} &= \gamma I
    \end{aligned}
\end{equation}

If the intervention is initiated at some time $t = t_i$, it must stop at time $t = t_i + \tau$. So necessarily $b(t) = 1$ if $t < t_i$ or $t > t_i + \tau$. During the intervention (i.e.~when $t_i \le t \le t_i + \tau$), $b(t)$ is an arbitrary function, possibly discontinuous, with range in $[0,1]$. This restriction assumes that we cannot intervene increase transmission above what occurs in the absence of intervention, but also optimistically assumes that $b(t)$ can be adjusted instantaneously, and that the effective reproduction number \Reff{} can be reduced all the way to zero, at least for a limited time. 

\subsection*{The optimal intervention}
We pose the following optimization problem: given the epidemiological parameters $\Ro$ and $\gamma$ and the finite duration $\tau$, what is the optimal intervention $b(t)$ that minimizes the epidemic peak \Imax? The optimal intervention is of interest for two reasons. It provides a reference point for evaluating alternative interventions. And it will allow us to analyze the inherent risks and shortcomings of time-limited interventions, even in the best-case scenario.

We prove (\snref{sec:supp-theorems}, Theorem \ref{theorem:maintain-then-suppress}) that for any $\Ro$, recovery rate $\gamma$, and duration $\tau$, there is a unique globally optimal intervention $b(t)$ that starts at an optimal time \tiopt{} and is given by:

\begin{equation}
\begin{aligned}
\label{eqn:optimal}
\bopt(t)=  \begin{cases}
    \frac{\gamma}{\beta S}, & t\in[t_{i},t_{i}+f\tau) \\
	0, & t\in[t_{i}+f\tau,t_{i}+\tau] \\
  \end{cases} 
\end{aligned}
\end{equation}

Where $f$ takes a specific value in $[0,1]$. This solution says that the optimal strategy is to maintain and then suppress. The intervention spends a fraction $f$ of the total duration $\tau$ in a maintain phase, with $b(t)$ chosen so that $\Reff = 1$. This maintains the epidemic at a constant number of infectious individuals equal to $I(\tiopt)$, while susceptibles are depleted at a rate $\gamma I(\tiopt)$. The intervention then spends the remaining fraction $1 - f$ of the total duration in a suppress phase, setting $\Reff = 0$ so that infectious individuals are depleted at a rate $\gamma I$ (Fig.~\ref{fig:intervention}a,d). 

The effectiveness of the optimal intervention, when it should commence, and the balance between maintenance and suppression all depend on the total allowed duration of the intervention, $\tau$. The longer we can intervene (larger $\tau$), the more we can reduce the peak (smaller $\Imax$), the earlier we should optimally act (smaller \tiopt), and the longer we should spend maintaining versus suppressing (larger $f$) (Fig.~\ref{fig:intervention}a,g,h, \snref{sec:supp-theorems}, Theorem \ref{theorem:maintain-then-suppress}, Lemma \ref{lemma:lots-of-powder-shoot-early}, and Corollary \ref{corr:bunker-hill}). Notably, there is a critical duration $\taucrit$ such that if the intervention is shorter than $\taucrit$, then the optimal strategy is to suppress for the entire duration of the intervention (Fig.~\ref{fig:intervention}g,h, Supplementary Fig.~\ref{fig:f-sigma-of-tau}, \snref{sec:supp-theorems}, Theorem \ref{theorem:always-suppress}). 

\begin{figure}[ht]
    \centering
\makebox[\textwidth][c]{    \includegraphics[width=7.5in]{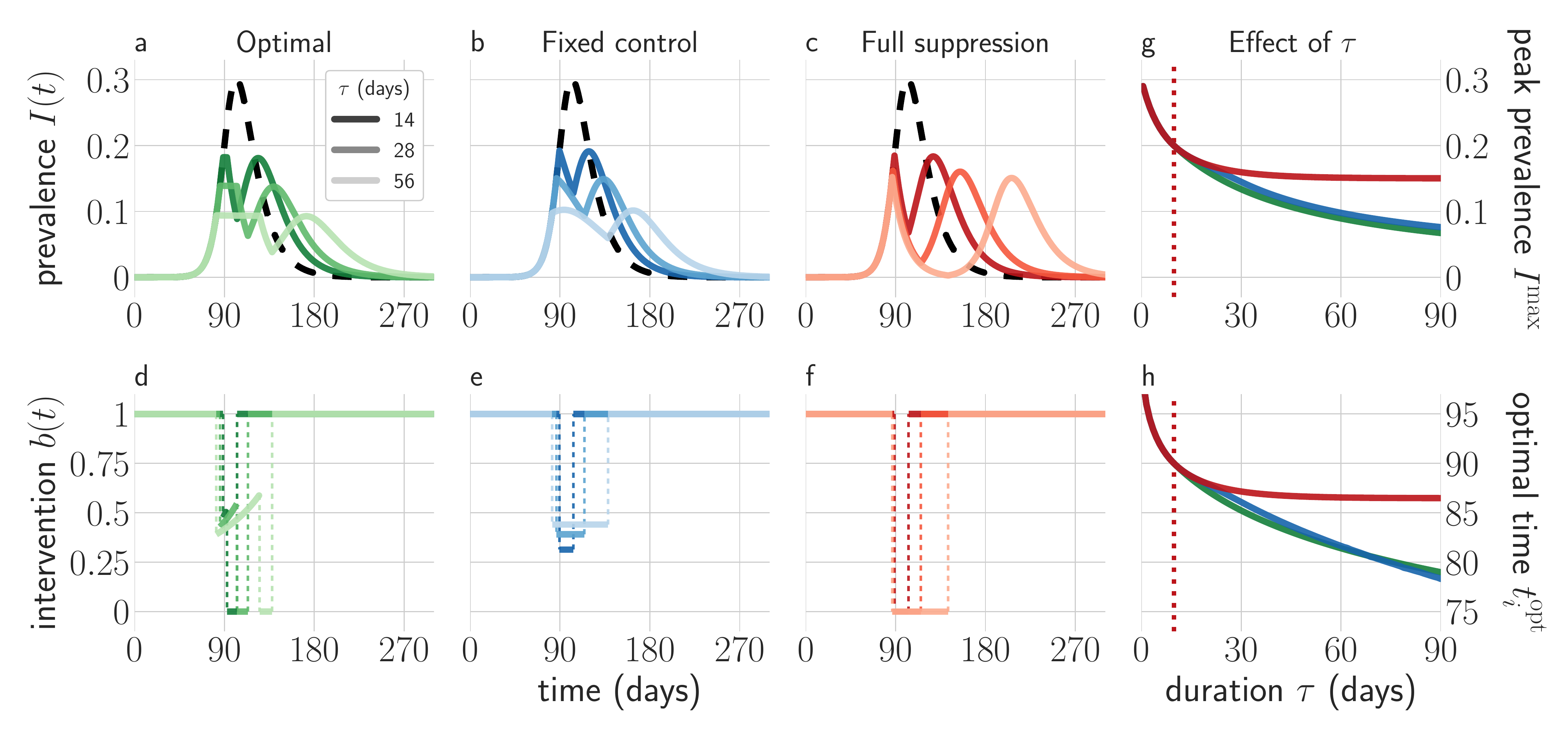}}
\autocaption{Interventions reduce peak infection prevalence.}{\textbf{a--f}, Timecourses of epidemics under optimal (\textbf{a}), fixed control (\textbf{b}), and full suppression (\textbf{c}) interventions with three different values of intervention duration $\tau$: 14 days (dark lines), 28 days (intermediate lines), and 56 days (light lines), and their respective intervention functions $b(t)$ (\textbf{d}--\textbf{f}). At time $t$, the basic reproduction number \Ro{} is reduced to some fraction of its maximum (no intervention) value; $b(t)$ gives that fraction. Dashed black line shows timecourse in the absence of intervention. \textbf{g, h}, Effect of the duration $\tau$ on the infectious prevalence peak (\textbf{g}) and intervention starting times (\textbf{h}) for the optimal intervention (green), the optimized fixed control intervention (blue), and the optimized full suppression intervention (red). Dotted red line shows the critical value $\taucrit$ below which full suppression is the globally optimal intervention. Plotted parameters: $\Ro = 3$, $\gamma = \frac{1}{14}$ days$^{-1}$. See Table \ref{tab:default-parameters} for parameter justifications.}
\label{fig:intervention}
\end{figure}
\clearpage

\subsection*{Near-optimal interventions}
Although theoretically enlightening, the optimal intervention described above is not feasible in practice. Implementing it would require policies flexible enough to fine-tune transmission rates continuously, imposing ever changing social behaviors. It would also require instantaneous and perfect information about the current state of the epidemic in the population, information that will clearly not be available even with greatly improved testing and contact-tracing. 

We therefore also consider other families of potential interventions, and study how they perform compared to the optimal intervention.

Real-world interventions typically consist of simple rules that are fixed for some period of time (quarantines, restaurant closures, physical distancing). We model such fixed control strategies (previously investigated by others \autocite{dilauro2020timing,hollingsworth2011mitigation}) as interventions of the form:

\begin{equation}
\begin{aligned}
\label{eqn:fixed_control}
b_\text{fix}(t)= \sigma \text{ for } & t\in[t_{i},t_{i}+\tau] \\
\end{aligned}
\end{equation}

Fixed control interventions are determined by two parameters: the starting time $t_i$ and the strictness $\sigma \in[0,1]$. For any intervention duration $\tau$ we can numerically optimize $t_i$ and $\sigma$ to minimize peak prevalence $\Imax$ (Fig.~\ref{fig:intervention}b,e)

For a given $\Ro$, $\gamma$, and $\tau$, an optimized fixed control intervention yields an epidemic time course that is remarkably similar to the one obtained under the globally optimal intervention strategy (Fig.~\ref{fig:intervention}a,b). Peak prevalence \Imax{} is only slightly lower in the optimal intervention than in an optimized fixed control intervention (Fig.~\ref{fig:intervention}a,b,g). The effectiveness of a fixed control intervention depends on $\tau$ in a similar manner to that of the optimal intervention: longer interventions are more effective, should start earlier, and are less strict (Fig.~\ref{fig:intervention}e,g,h).

The similarities between fixed control interventions and optimal interventions can be understood by considering the time course of $\Reff$ during the intervention. At first, the fixed control intervention mainly depletes the susceptible fraction. As susceptible hosts are depleted, $\Reff$ falls, so the intervention naturally begins to deplete the infectious fraction. And so fixed control interventions are qualitatively similar to the maintain, then suppress optimal intervention. Optimizing a fixed control strategy has the effect of choosing a $\sigma$ and $t_i$ that emulate---and thus perform nearly as well as---an optimal intervention (Fig.~\ref{fig:intervention}a,b,g).

As $\tau$ becomes small, the optimal $\sigma$ for a fixed control intervention also becomes small, mimicking the suppression phase of the optimal intervention. For small enough $\tau$, the optimal $\sigma$ is equal to 0: the intervention consists entirely of suppression (Supplementary Fig.~\ref{fig:f-sigma-of-tau}a). As $\tau$ increases, the optimal $\sigma$ also increases, producing interventions with longer and longer susceptible-depleting maintenance-like periods (Supplementary Fig.~\ref{fig:f-sigma-of-tau}a). The epidemic trajectory becomes increasingly flat, eventually approximating a pure maintenance intervention. Note that achieving this approximate flatness with a fixed control intervention requires allowing some initial growth of the infected class ($\Reff > 1$ at $t_i$), albeit at a reduced rate compared to no intervention. Growth of the infectious class never occurs during optimal interventions.

We also analyze a third class of interventions: the full suppression interventions defined by

\begin{equation}
\begin{aligned}
\label{eqn:full_supression}
b_{0}(t)= 0 \text{,  } & t\in[t_{i},t_{i}+\tau]. \\
\end{aligned}
\end{equation}

These interventions emulate extremely strict quarantines (Fig.~\ref{fig:intervention}c,f). They are characterized by the complete absence of susceptible depletion. Such interventions are fully determined by the starting time $t_i$, which can be optimized given the total allowable duration $\tau$. 

Note that full suppression interventions are a limiting case both of maintain-suppress interventions (with no maintenance phase, $f = 0$) and of fixed control interventions (with maximal strictness, $\sigma=0$). Accordingly, the optimized full suppression intervention performs similarly to the optimal intervention and to the optimized fixed control intervention for short durations, when those favor a relatively short maintenance phase and high strictness (Fig.~\ref{fig:intervention}c). For longer interventions, the effectiveness of full suppression rapidly plateaus (Fig.~\ref{fig:intervention}g, \snref{sec:supp-theorems}, Corollary \ref{corr:full_supression_optima}). Accordingly, the optimal time to initiate a full suppression intervention plateaus with increasing $\tau$: there is no benefit in fully suppressing too early (Fig.~\ref{fig:intervention}h, \snref{sec:supp-theorems}, Corollary \ref{corr:full_supression_optima}).

Taken together, these results show that the most efficient way for long interventions to decrease peak prevalence \Imax{} is to cause susceptible depletion while limiting how much the number of infectious individuals can grow. For short interventions, by contrast, it is most efficient simply to reduce the number of infectious individuals. 
Optimizing an intervention trades off cases now against cases later. We prove (\snref{sec:supp-theorems}, Theorems \ref{theorem:twin-peaks}, \ref{theorem:twin-peaks-fixed}) that optimal and near-optimal interventions cause the epidemic to achieve the peak prevalence exactly twice: once during the intervention and once strictly afterward. All optimized maintain-suppress interventions of duration $\tau$ and maintenance fraction $f$ have this twin-peak property, including the globally optimal intervention, the optimized full suppression intervention, and also the optimized fixed control intervention for any duration $\tau$. Note that this means that the optimized interventions always end before herd immunity is reached. In practice this means that if an intervention continues until herd immunity is reached, then an earlier intervention of the same duration could have produced a lower epidemic peak while permitting a small rebound. The absence of a second peak is not a sign of policy success; it is a sign that policymakers acted too late.

\subsection*{Mistimed interventions}\label{sec:mistimed}
The optimal and near-optimal interventions are extremely powerful. For COVID-like epidemic parameters, the 28-day optimal or fixed control interventions reduce peak prevalence from about 30\% of the population to under 15\%. Even the full suppression intervention reduces peak prevalence to well under 20\%. These are massive and potentially health system-saving reductions.

In practice, however, interventions are not automatically triggered at a certain number of infectious individuals or at a certain point in time. They are introduced by policymakers, who must estimate the current quantity of infectious individuals $I(t)$, often from very limited data, must begin roll-out with an uncertain period of preparation, and must also estimate the epidemiological parameters $\Ro$ and $\gamma$. These tasks are difficult, and so policymakers may fail to intervene at the optimal moment \tiopt{}. We consider timing errors in both directions, though in practice acting late might be more common than acting early.

How costly is mistiming a time-limited intervention? We find that even a single week of separation between the time of intervention and $\tiopt$ can be enormously costly, for realistic COVID-19 epidemic parameters (Fig.~\ref{fig:mistimed}a--c, g--i).

\begin{figure}[ht]
    \centering
    \makebox[\textwidth][c]{
    \includegraphics[width=\textwidth]{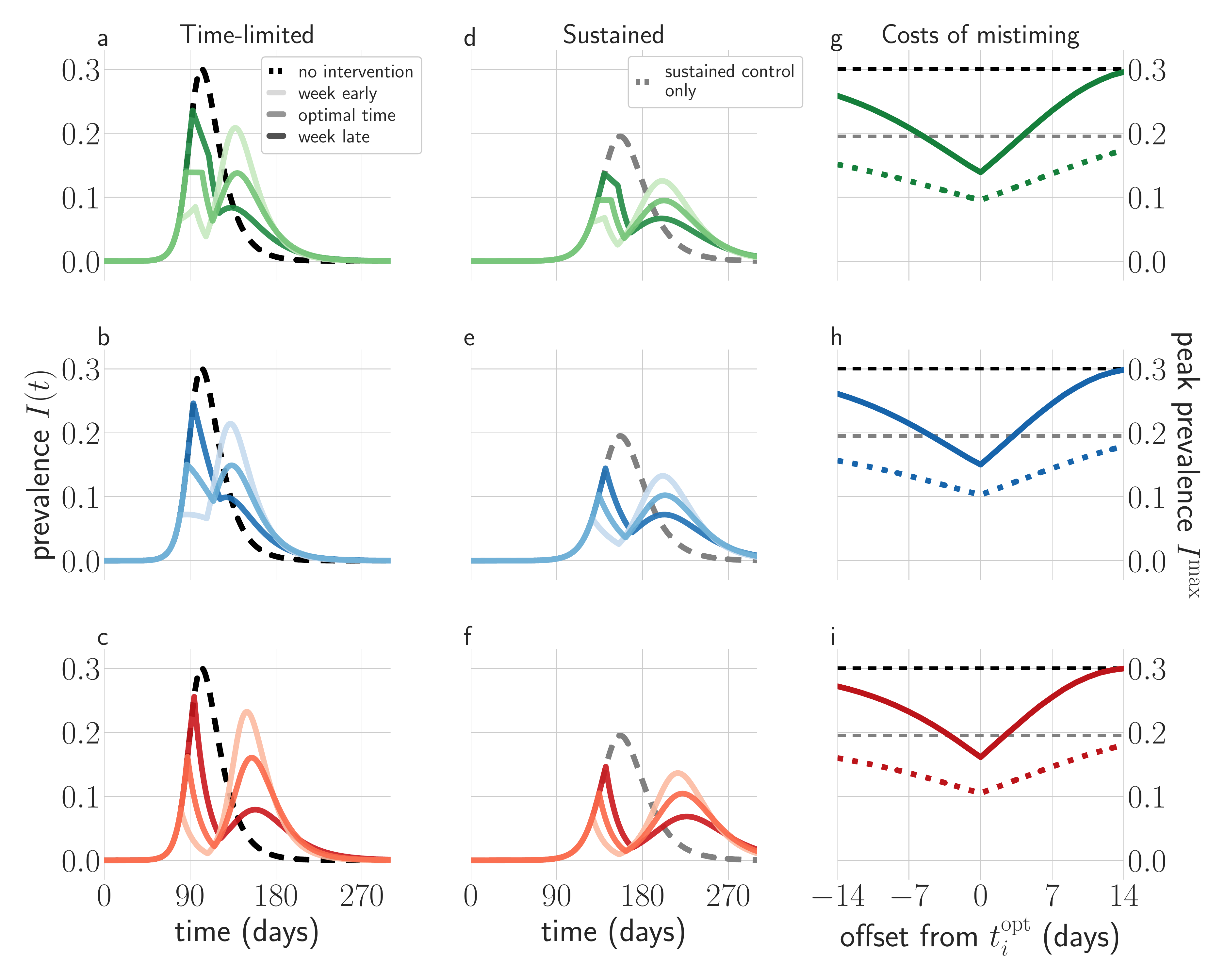}}
    \autocaption{Mistiming an intervention reduces its effectiveness.}{\textbf{a-f,} Timecourses of epidemics under optimal \textbf{a,} fixed control \textbf{b,} and full suppression \textbf{c,} interventions that are possibly mistimed: a week late (dark lines), optimally timed (intermediate lines), and a week early (light lines). Dashed black line shows timecourse in the absence of intervention. \textbf{d-f,} Timecourses of epidemics with a sustained control that reduces the basic reproduction number $\Ro$ by $25$\%, combined with the effects of a possibly mistimed optimal \textbf{d,} fixed control \textbf{b,} and full suppression \textbf{f} interventions, with line lightness as before. Dashed grey line shows timecourse with only sustained control and no additional intervention. \textbf{g-i,} Effect of offset of intervention time $t_i$ from optimal intervention time $\tiopt{}$ on epidemic peak prevalence \Imax{} without (solid lines) and with (dotted lines) sustained control for optimal \textbf{g,} fixed control \textbf{h,} and full suppression \textbf{i} interventions. Dashed black and grey lines show \Imax{} in the absence of intervention, without and with sustained control, respectively. Plotted parameters: $\Ro = 3$, $\gamma = \frac{1}{14}$ days$^{-1}$. See Table \ref{tab:default-parameters} for parameter justifications.}
    \label{fig:mistimed}
\end{figure}

While the optimal intervention achieves a dramatic reduction in the height of the peak, mistiming such an intervention can be disastrous. Intervening too early produces a resurgent peak, but it is even worse to intervene too late. For example, if the intervention is initiated one week later than the optimal time, then \Imax{} is barely reduced compared to the absence of any intervention whatsoever, particularly for a full suppression intervention (Fig.~\ref{fig:mistimed}a--c,g--i).

The extreme costs of mistiming arise from the steepness of the $I(t)$ curve at \tiopt. Optimized interventions permit some cases now in order to reduce cases later. Both infectious depletion and susceptible depletion require currently infectious individuals in order to be effective at reducing peak prevalence. This means that, except for interventions of very long duration $\tau$, the optimal start time $\tiopt{}$ occurs during a period of rapid, near-exponential growth in the fraction infectious $I(t)$. 

The practical problem with this approach is intuitive: because $S(t)$ and $I(t)$ are so steep at $S(\tiopt{}), I(\tiopt{})$, small errors in timing produce large errors in terms of $S(t_i), I(t_i)$ (Fig.~\ref{fig:mistimed}a--c). Indeed, for epidemics that have faster dynamics, the consequences of mistiming interventions are increasingly stark (Fig.~\ref{fig:mistimed}a--f, Fig.~\ref{fig:parameter-sweep}b). 

It is also clear why being late is costlier than being early (Fig.~\ref{fig:mistimed}g--i, Fig.~\ref{fig:parameter-sweep}a--d). An early intervention is followed by a large resurgent second peak, but the resurgence is slower and smaller than the uncontrolled initial surge permitted by a late intervention, thanks to the susceptible depletion that occurs during the early intervention (Fig.~\ref{fig:mistimed}a--c,g--i, Fig.~\ref{fig:parameter-sweep}a--d).

Importantly, \snref{sec:supp-theorems} Theorems \ref{theorem:twin-peaks} and \ref{theorem:twin-peaks-fixed} imply that late implementation of an optimized intervention results in a maximal epidemic peak at the time the intervention is initiated; whereas early implementation postpones the maximal peak (Fig.~\ref{fig:mistimed}a--c). In practice this means that early interventions allow for course corrections.

\begin{figure}[ht]
    \centering
    \makebox[\textwidth][c]{
    \includegraphics[width=\textwidth]{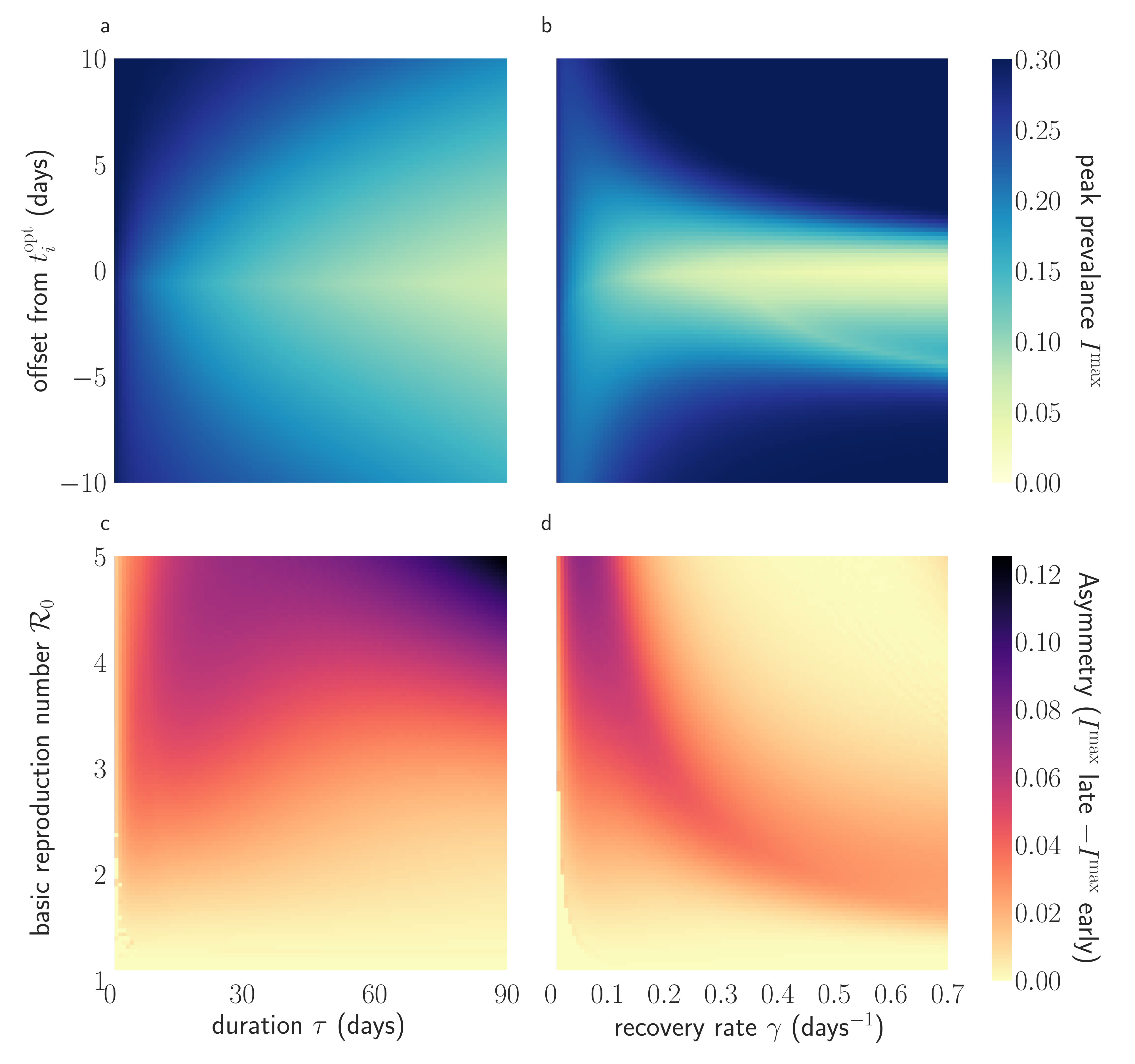}}
    \autocaption{Effect of parameter variation on peak reduction and robustness to mistiming.}{\textbf{a, b}, Peak prevalence $\Imax$ as a function of offset from optimal intervention initiation time $\tiopt{}$ and (\textbf{a}) intervention duration $\tau$, (\textbf{b}) recovery rate $\gamma$. \textbf{c, d}, Asymmetry between intervening early and intervening late, quantified as \Imax{} for a 7-day late intervention minus \Imax{} for a 7-day early intervention, as a function of the basic reproduction number $\Ro$ and (\textbf{c}) duration $\tau$, (\textbf{d}) recovery rate $\gamma$. Unless otherwise stated, $\Ro = 3$, $\gamma = \frac{1}{14}$ days$^{-1}$, $\tau = 28$ days. See Table \ref{tab:default-parameters} for parameter justifications.}
    \label{fig:parameter-sweep}
\end{figure}
\clearpage

\subsection*{Sustained interventions}
Are there any sustained measures that can improve the robustness of optimized interventions? Because the severity of mistiming is governed by the steepness of $I(t)$, measures that reduce steepness should alleviate the impact of mistiming. We propose using weak measures of long duration to achieve this desired outcome. Even though these measures by themselves may have little effect on $\Imax$, they can buffer timing mistakes when used in combination with stronger, time-limited interventions.

We consider sustained weak interventions, modeled as a constant reduction of $\Ro$ throughout the entire epidemic, both on their own and combined with optimized time-limited interventions (Fig.~\ref{fig:mistimed}d--f). If perfectly timed, an optimized time-limited intervention outperforms a sustained intervention. Moreover, adding a sustained intervention to a perfectly-timed time-limited intervention provides little extra benefit (Fig.~\ref{fig:mistimed}g--i). However, even a slightly mistimed time-limited intervention is far worse than a sustained intervention. And, most importantly, if both sustained and time-limited interventions are adopted, the time-sensitivity of the time-limited intervention is reduced. In particular, the otherwise disastrous cost of intervening too late is reduced (Fig.~\ref{fig:mistimed}d--i).
\clearpage
\section*{Discussion}
COVID-19 has thrown into relief the importance of flattening the curve. Our analysis establishes the optimal strategy to minimize peak prevalence in the SIR epidemic model, given an intervention of limited duration. Simpler interventions can closely approximate the optimal outcome.  

Deriving the optimal strategy highlights the fundamental materials---susceptible depletion and infectious depletion---of any epidemic mitigation strategy, and it provides a yardstick against which to measure all other strategies.

But it would be unwise to attempt an optimal or near-optimal intervention in practice. The inevitable errors in timing that arise from uncertainty in inference and delays in implementation will produce disaster.

It is particularly costly to act too late. This causes an elevated peak immediately before the intervention even starts. By contrast, a premature intervention leads to a substantially delayed second wave \autocite{leung2020first} after the relaxation of controls. Such a delay may sometimes be more desirable than peak reduction itself. A second wave is less challenging to control and manage than a first wave: healthcare capacity can be increased in the interim, pharmaceutical interventions such as antivirals may become available, epidemiological parameters will be better known, and accumulated population immunity will reduce exponential growth rate even in the absence of intervention. Moreover, our results apply to any interventions that policymakers might be able to take during a second wave. 

Our analysis has many limitations, and it leaves a large body of important questions unresolved. We have generously assumed that policymakers possess complete information about epidemic parameters and about the current epidemic state. Future studies can work to address the problem of disease control under uncertainty, by coupling an understanding of optimal interventions with an analysis of epidemiological inference. Errors of inference are likely to be magnified by subsequent intervention mistiming. Moreover, our analysis makes the unrealistic assumption that intervention strength $b(t)$ can be tuned at will. But in practice $b(t)$ can be tuned only coarsely, and even fixed control interventions are not truly enforceable in their idealized form.

We have on time-limited interventions. This time limit captures the social and economic costs of intervention. It also reflects the possibility that non-compliance with measures may rise over time, a concern that some policymakers had when planning COVID-19 responses \autocite{harvey2020behavioral}.

We have focused our analysis on the epidemic peak. This quantity is critical because it is the point at which health services will be most strained. An overwhelmed health system can dramatically increase infection-fatality rates, direct morbidity, and medical complications \autocite{kissler2020social,ferguson2020impact}. Peak epidemic prevalence is a good proxy for demands on the healthcare system, but an even better metric is the total person-days with prevalence exceeding the maximum healthcare capacity. If the peak prevalence cannot be reduced below this capacity, then the strategy that minimizes the peak will not necessarily minimize the cumulative impact above capacity. It might be preferable to permit a slightly higher peak, and then move to full suppression. 

A policymaker also seeks to reduce the total cases during the epidemic. But for an explosively spreading novel pathogen, this consideration may be secondary to reducing peak prevalence and avoiding healthcare system collapse. Moreover, interventions that reduce the epidemic peak almost necessarily reduce the total case count (or final size) of the epidemic, though they may not do so as efficiently as interventions targeted directly at final size reduction \autocite{dilauro2020timing}. A policymaker may also wish to delay the epidemic peak to allow time for healthcare capacity to build and pharmaceuticals to be developed; this too produces different policy trade-offs \autocite{dilauro2020timing}.

In general, however, when healthcare capacity is larger, the problem of minimizing the cumulative impact above capacity becomes almost identical to the problem of minimizing peak prevalence. Whereas when capacity is smaller, the problem of minimizing cumulative impact becomes the same as minimizing the final size of the outbreak \autocite{dilauro2020timing}. 

We use one of the simplest possible models of disease transmission: the SIR model in a homogeneously mixing population, and in the large population limit in which differential equations are appropriate. This is by design. We show that even in a simple setting, without the other confounding factors of real world disease spread, such as population structure, stochasticity, time varying transmission rates, or partial immunity, time-limited non-pharmaceutical control is not robust to implementation error.

Had our model included those complicating factors, one might have conjectured that the danger of mistiming is linked to those modeled factors chosen, and could therefore be controlled by carefully accounting for them. Had we studied an intervention that was not provably optimal, one could have argued that a more optimal intervention would not carry such risks. Instead, the simplicity of our model and the provable optimality of the strategies studied demonstrate that these risks of mistiming are a fundamental feature of the initial exponential growth of an epidemic, regardless of the optimality of the intervention or the putative realism of the model used.

That said, time-limited non-pharmaceutical control in less idealized circumstances is worthy of study. di Lauro and colleagues provide a numerical treatment of a metapopulation of internally well-mixed SIR demes \autocite{dilauro2020timing}. Countries such as Vietnam \autocite{lee2020should}, Taiwan \autocite{summers2020potential}, and New Zealand \autocite{baker2020successful} successfully adopted non-pharmaceutical strategies aimed at eliminating of SARS-CoV-2 transmission locally and then suppressing reintroductions. Metapopulation models could reveal the conditions under which a local elimination with a time-limited intervention is viable and robust. Similarly, heterogeneities in a single population---for instance in individual susceptibility or in degree of social connectedness---can alter disease dynamics \autocite{miller2009spread,volz2011effects}. Future studies could assess the impact of realistic network topologies on the optimality and robustness of interventions.

Still, our simple analysis offers several clear, practical principles for policymakers. First, act early. There is a striking asymmetry in the costs of acting too early versus too late. Second, work to slow things down. Slowing the growth curve makes interventions more robust, and also makes inference of epidemiological parameters more accurate. Third, when in doubt, bear down. Even the crude policy of full suppression is remarkably successful at reducing peaks and delaying excess prevalence. And if policymakers are very late to act then full suppression is, in fact, optimal. 

Naive optimization is dangerous. Real-world policy must emphasize robustness, not efficiency.

\clearpage
\section*{Methods}

\subsection*{Model parameters}
Table \ref{tab:default-parameters}, below, gives definitions of model parameters, their units, default values plotted in figures unless otherwise stated, and justifications for those choices.

\begin{table}[h]
\centering
\scriptsize
\renewcommand{\arraystretch}{2}
\caption{Model parameters, default values, and sources/justifications}
\begin{tabular}{ l p{3.5cm} l l  p{4cm}}
Parameter & Meaning & Units & Value & Source or justification \\
\hline
$\Ro$ & basic reproduction number & unitless & 3 & Estimates for COVID range from 2 to 3.5 \autocite{li2020substantial,park2020reconciling}\\
$\gamma$ & recovery rate &$1/$days & $\frac{1}{14}$ & Infectious period for COVID of approximately 1--2 weeks \autocite{zou2020sars} \\
$\beta$ & rate of disease-causing contact &$1/$days & $\Ro \gamma$ & calculated  \\
$\tau$ & duration of a time-limited intervention & days & 28 & Approximately one month \\
\hline
\end{tabular}
\label{tab:default-parameters}
\end{table}
\clearpage
\subsection*{State-tuned and time-tuned maintain-suppress interventions}
When we ask what it means for a maintain-suppress intervention to be mistimed, we need a model of how the intervention is implemented. One possibility is that the policymaker directly observes $S(t)$ throughout the intervention and chooses $b(t) = \frac{\gamma}{\beta S(t)}$ during the maintenance phase based on the directly observed $S(t)$. We call this a state-tuned intervention.

Alternatively, the policymaker plans to intervene at some $S(t)$ value $S_i$ predicted to occur at a time $t_i$. The policymaker knows that during a successful maintenance phase, $S(t) = S_i - \gamma I_i (t - t_i)$. The policymaker then chooses the maintenance phase values of $b(t)$ according to this predicted $S(t)$. We call this a time-tuned intervention. 

When we study mistimed interventions, we use time-tuned interventions. Since instantaneous epidemiological observation is not possible, time-tuned interventions are a more realistic model of how a maintain-suppress intervention, if possible at all, would in fact be implemented. If instantaneous epidemiological observation were possible during interventions, timing errors could be mitigated better than either state-tuning or time-tuning allows. Policymakers could observe the true $(S_i, I_i)$ at the moment of intervention $t_i$ and then employ whichever intervention of duration $\tau$ is optimal given that it begins at $(S_i, I_i)$.

Indeed, it can be seen that time-tuned interventions are in fact slightly more robust to mistiming than state-tuned interventions. They are partially self-correcting where the state-tuned interventions are not. 

Late interventions have $\Imax = I_i$. But since a late time-tuned intervention has higher initial strictness than a late state-tuned intervention, it avoids unnecessary time spent at $\Imax = I_i$.

Early time-tuned interventions achieve lower \Imax{} than equivalent early state-tuned interventions. This is true even---in fact, especially---for very fast epidemics. During the maintain phase of a maintain-suppress intervention, the strictness decreases in time (see Supplementary Note \ref{sec:supp-theorems}, equation \ref{eqn:bopt_time}). If the intervention is too early, this allows for some initial growth of the infectious fraction and thus improved depletion of the susceptible fraction (Fig.~\ref{fig:mistimed}a,d) relative to maintaining at $I_i$ (as in a state-tuned intervention). The ensuing second peak is therefore reduced. 

An intriguing side effect of this automatic course-correction is that for relatively fast epidemics, some premature interventions outperform others that are less early (note the branching in Fig.~\ref{fig:parameter-sweep}b).

\printbibliography[section = 1]
\section*{Acknowledgements}
We thank Ada W.\ Yan, Amandine Gamble, Corina E.\ Tarnita, Elizabeth N.\ Blackmore, James O.\ Lloyd-Smith, and Judith Miller for helpful comments on previous versions of this work. We thank Juan Bonachela for helpful discussions. DHM and SAL gratefully acknowledge financial support from NSF grant CCF 1917819 and a C3.ai DTI Research Award from C3.ai Inc.\ and Microsoft Corporation. SAL gratefully acknowledges financial support in the form of a gift from Google, LLC.\ for work on COVID-19.

\section*{Author contributions}
DHM conceived the study. DHM and FWR designed and analyzed the mathematical model, with support and proof verification from JBP and SAL. FWR proved key theorems, with support from DHM. DHM conducted numerical analysis and produced figures, with support from FWR. DHM, FWR, and JBP wrote the first draft of the manuscript, which all authors edited.

\section*{Competing interests}
We have no competing interests to declare.

\section*{Additional information}

\subsection*{Supplementary Information}
Supplementary information for this paper is available online.

\subsection*{Correspondence}
Correspondence and requests for materials should be addressed to \href{mailto:dylan@dylanhmorris.com}{Dylan H. Morris} and \href{mailto:frossine@princeton.edu}{Fernando W. Rossine}.

\subsection*{Data availability}
Data sharing is not applicable to this article as no datasets were generated or analysed during the current study. Code to reproduce numerical model analysis is provided (see Code availability below).

\subsection*{Code availability}\label{sec:code-avail}
All code needed to reproduce numerical results and figures is archived on Github (\githuburl) and on OSF (\osfurl), and licensed for reuse, with appropriate attribution/citation, under a BSD 3-Clause Revised License.

We wrote numerical analysis and figure generation code in Python 3 \autocite{python}, using numerical solvers provided in NumPy \autocite{numpy} and SciPy \autocite{scipy}, and produced figures using Matplotlib \autocite{matplotlib}. Parameter choices for numerical analysis are stated in figure captions and in Supplementary Table \ref{tab:default-parameters}.
\end{refsection}

\begin{refsection}
\clearpage
\appendix 

\setcounter{figure}{0}
\renewcommand{\thefigure}{\arabic{figure}}

\captionsetup[figure]{
labelfont={bf}, 
labelformat = mid,
labelsep = none,
name = {Supplementary Fig.}}

\renewcommand{\thesection}{\arabic{section}}

\renewcommand{\thesubsection}{\arabic{section}.\arabic{subsection}}

\setcounter{section}{0}
\setcounter{page}{1}

\title{\textbf{\SupplementNameFull{} for}\\ \papertitle}

\maketitle

\clearpage

\section*{Supplementary Figures}
\begin{figure}[ht]
    \centering
    \makebox[\textwidth][c]{
    \includegraphics[height=0.5\textheight]{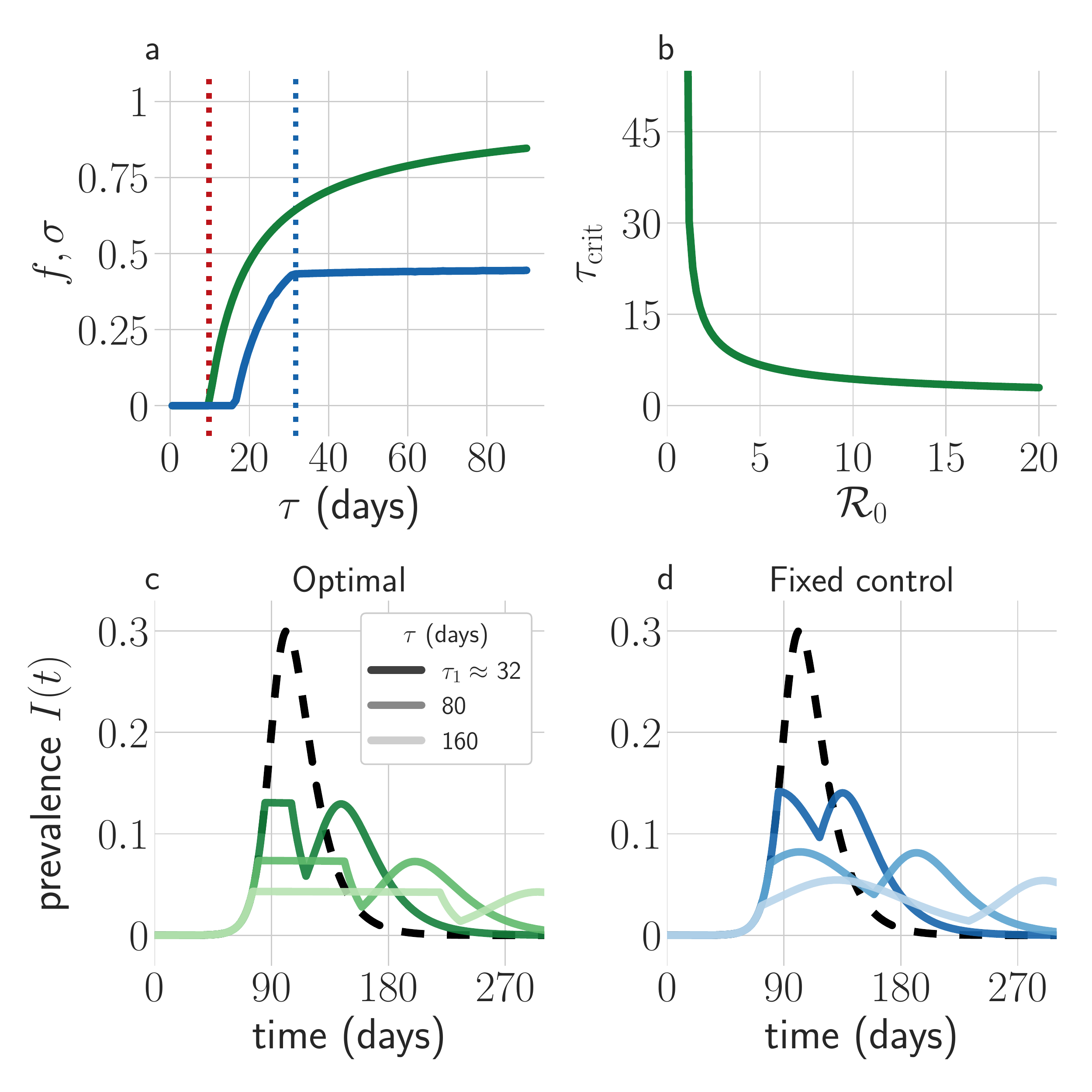}}
    \autocaption{Effect of $\tau$ and $\Ro$ on the optimal and optimized fixed control strategies.}{\textbf{a,} Optimal maintenance fraction $f$ for the optimal strategy (green line) and optimized strictness $\sigma$ for the fixed control intervention (blue line) as a function of intervention duration $\tau$. Red dotted vertical line shows the critical value $\taucrit$ below which full suppression ($f=0$) is the globally optimal intervention. Blue dotted vertical line shows $\tau_1$, the point at which $\Reff = 1$ at the start of the optimized fixed control intervention. \textbf{b,} Full suppression critical value $\taucrit$ as a function of basic reproduction number \Ro. \textbf{c, d,} Example timecourses for $\tau \ge \tau_1$, under the (\textbf{c}) optimal and (\textbf{d}) fixed control interventions. Dashed black line shows timecourse in the absence of intervention. Parameters as in Table \ref{tab:default-parameters} of the main text unless otherwise stated.}
    \label{fig:f-sigma-of-tau}
\end{figure}
\clearpage
\section{Supplementary Note 1: Theorems and Proofs}\label{sec:supp-theorems}

\subsection{Useful notation}
We define \Scrit{} for an SIR model to be the critical fraction susceptible at which $\Reff = 1$ and $\dv{I}{t} = 0$ (in the absence of intervention), i.e. $\Scrit \equiv \frac{1}{\Ro}$. 

\subsection{Peak prevalence  in an SIR model}
Define $\Imax(t)$ for an SIR system as the maximum value of $I(x)$ achieved during the interval $x \in [t, \infty)$. Notice that $\Imax(0) = \Imax$, where $\Imax$ is the global maximum value of $I(x)$, which we are seeking to minimize with our intervention.

A known result that is immediate from the original work of Kermack and McKendrick \autocite{kermack1927contribution} (see Weiss \autocite{weiss2013sir} for an explicit derivation) holds that for $S(t) \ge \Scrit$:
\begin{equation}
\label{eqn:Imax-t}
\Imax(t) = 
I(t) + S(t) - \frac{1}{\Ro} \log\Big(S(t)\Big) - \frac{1}{\Ro} + \frac{1}{\Ro} \log\Big(\frac{1}{\Ro}\Big)
\end{equation}

\begin{remark}[Continuity of $\Imax$]
\label{remark:continuity-of-imax-t}
$\Imax$ is a continuous function $\Imax: (0, 1]^2 \mapsto [0, 1]$ of $v = (S(t), I(t))$ for $ v \in [\Scrit, 1] \cross (0, 1]$, and therefore also a continuous function of $t$, $\Imax: \mathbb{R} \mapsto [0, 1]$ for $t \in (-\infty, t_\text{crit}]$, where $t_\text{crit}$ is the time such that $S(t_\text{crit}) = \Scrit$ \end{remark} 
This is immediate from the fact that $\Imax(t)$ is a linear combination of univariate functions of $S(t)$ and $I(t)$ that are themselves continuous on $[\Scrit, 1]$ and $[0, 1]$, respectively. And since $S(t)$ and $I(t)$ are continuous functions of $t$, \Imax{} is a continuous function of $t$ for $t \in (-\infty, t_\text{crit}]$.

\begin{lemma}[The more fire, the bigger the blaze]
\label{lemma:fire-fire}
If $0 \le I_x(t_x) \le I_y(t_y)$ and $S_x(t_x) = S_y(t_y)$ for two SIR systems $x$ and $y$ with identical parameters \Ro{} and $\gamma$ at possibly distinct times $t_x, t_y \in [0, \infty]$ then $\Imax_x(t_x) \le \Imax_y(t_y)$, with equality only if $I_x(t_x) = I_y(t_y)$.
\end{lemma}

\begin{proof}
There are three cases.

\textbf{Case 1:} $I_x(t_x) = 0$. In this case, $\Imax_x(t_x) = 0 \le  I_y(t_y) \le \Imax_y(t_y) $, with equality only if $\Imax_y(t_y) = 0$, which can only occur if $I_y(t_y) = 0$.

\textbf{Case 2:} $S_x(t_x) = S_y(t_y) < \Scrit$. In this case, $\Imax_x(t_x) = I_x(t_x)$ and $\Imax_y(t_y) = I_y(t_y)$, so our result holds.

\textbf{Case 3:} $S_x(t_x) = S_y(t_y) > \Scrit$. In this case, we can apply equation \ref{eqn:Imax-t}. Fixing $S(t) > \Scrit$, $\Imax(t)$ is an increasing function of $I(t)$ (there is a single, positive $I(t)$ term in the sum), so the result must hold.
\end{proof}

\begin{lemma}[The more fuel, the bigger the blaze]
\label{lemma:fuel-fire}
If $I_x(t_x) = I_y(t_y) > 0$ and $S_x(t_y) \le S_y(t_y)$ for two SIR systems $x$ and $y$ with identical parameters \Ro{} and $\gamma$ at possibly distinct times $t_x, t_y \in [0, \infty]$ then $\Imax_{t_x} \le \Imax_{t_y}$, with equality only if $S_x(t_x) = S_y(t_y)$ or $S_y(t_y) \le \Scrit$
\end{lemma}

\begin{proof}
There are three cases.

\textbf{Case 1:} $S_y(t_y) \le \Scrit$.
If $S_y(t_y) \le \Scrit$, then $S_x(t_x) \le \Scrit$, and neither epidemic will grow after $t_x$ or $t_y$, respectively. It follows that $\Imax_{t_x} = I_x(t_y) = \Imax_{t_y} = I_y(t_y)$

\textbf{Case 2:} $S_y(t_y) > \Scrit$, $S_x(t_x) < \Scrit$.
In this case, $\Imax_{t_x} = I_x(t_y)$ but $\Imax_{t_y} > I_y(t_y)$, since $\dv{I_y}{t} > 0$ if $I_y(t) > 0$ and $S_y(t) > \Scrit$. So we have $\Imax_{t_x} < \Imax_{t_y}$

\textbf{Case 3:} $S_x(t_x), S_y(t_y) > \Scrit$.
The result in this case follows immediately from the fact that, fixing $I(t)$, $\Imax(t)$ is an increasing function of $S(t)$ when $S(t) > \Scrit$. We can see that it is by taking the partial derivative with respect to $S$ of the expression in equation \ref{eqn:Imax-t}:
\begin{equation}
    \pdv{}{S}\Imax(t) = 1 - \frac{1}{\Ro S}
\end{equation}

If $S > \Scrit$, $\frac{1}{\Ro S} < 1$, so  $\pdv{}{S}\Imax(t) > 0$, and $\Imax(t)$ is an increasing function of $S(t)$.
\end{proof}

\subsection{Intervention function}\label{sec:intervention-formal}
We say that a right-continuous function with finite discontinuity points $b(t): \mathbb{R} \mapsto \left [ 0,1 \right ]$ is an intervention beginning at $t_{i}$ with duration $\tau$ if $b(t)\equiv 1$ for all $t \in (-\infty,t_{i}  )\bigcup \left(t_{i}+\tau,\infty\right)$. The SIR model under such an intervention will then take the form

\begin{equation}
\begin{aligned}
\frac{dS}{dt} &= -b(t)*\beta SI\\
\frac{dI}{dt} &= b(t)*\beta SI-\gamma I\\
\frac{dR}{dt} &= \gamma I\\
\end{aligned}
\end{equation}

We wish to show that for every $\tau$ there exists an intervention that minimizes \Imax{} and that such an optimal intervention must be identical except at a finite set of times $\{t_j \mid t_j \in (t_i, t_i + \tau)\}$ to one of the form:
\begin{equation}
\begin{aligned}
\bopt(t)=  \begin{cases}
    \frac{\gamma}{\beta S}, & t\in[t_{i},t_{i}+f\tau) \\
	0, & t\in[t_{i}+f\tau,t_{i}+\tau] \\
  \end{cases}
\end{aligned}
\end{equation}
For some values of $t_{i}\in  \mathbb{R}$ and $f\in (0,1]$. We divide the proof into a series of lemmas.

\begin{lemma}[More fire, less fuel]
\label{lemma:more_infected_less_susceptible}
Let $b_{x}(t)$ and $b_{y}(t)$ be two interventions beginning at the same $t_{i}$ and lasting until $t_{f}=t_{i}+\tau$. Denote the infectious and susceptible fractions for $b_x(t)$ and $b_y(t)$, respectively, by $I_x(t), S_x(t)$ and $I_y(t), S_y(t)$. If $I_{x}(t) \geq I_{y}(t)$ for all $t\in [t_{i},t_{f}]$, and $I_{x}(t) > I_{y}(t)$ for some $t\in (t_{i},t_{f})$, then $S_{x}(t_{f}) < S_{y}(t_{f})$.
\end{lemma}

\begin{proof}
The difference $\Delta R_{xy}(t)=R_{x}(t)-R_{y}(t)$ obeys $\frac{d\Delta R_{xy}}{dt}=\gamma (I_{x}-I_{y})$, which is non-negative for all $t\in [t_{i},t_{f}]$. Therefore $\Delta R_{xy}(t)$ is non-decreasing during any time interval contained in $[t_{i},t_{f}]$. Take $\bar{t}\in (t_{i},t_{f})$ such that $I_{x}(\bar{t})-I_{y}(\bar{t}) > 0$. Because the intervention $b(t)$ allows only a finite number of discontinuities, the resulting $I(t)$ is continuous. And so there is an $\epsilon$ such that $I_{x}(t)- I_{y}(t) >0$ for all $t\in [\bar{t}-\epsilon,\bar{t}+\epsilon]$. Therefore $\Delta R_{xy}(t)$ must be strictly increasing during $[\bar{t}-\epsilon,\bar{t}+\epsilon]$.

Because $\Delta R_{xy}(t)$ never decreases and sometimes increases during $[t_{i},t_{f}]$, it follows that $\Delta R_{xy}(t_{i})<\Delta R_{xy}(t_{f})$, but $\Delta R_{xy}(t_{i})=0$ so $\Delta R_{xy}(t_{f})>0$ and $R_{x}(t_{f})>R_{y}(t_{f})$. But $I_{x}(t_{f})\geq I_{y}(t_{f})$ and $S=1-(I+R)$, so $S_{x}(t_{f}) < S_{y}(t_{f})$.
\end{proof}

We now wish to prove that an epidemic left alone---without any intervention---will reach any given level of prevalence faster than it would have if any transmission-reducing intervention had taken place. In other words, an intervention always slows down the epidemic.

\begin{lemma}[Do nothing to burn faster and brighter]
\label{lemma:do_nothing}
Let $b_{1}(t) \equiv 1$ denote the null intervention and take $t_{max}$ such that $I_{1}(t_{max})=I_{1}^{max}$. Let $b_{x}(t)$ be an intervention that starts at time $t_i$, has duration $\tau$, and satisfies $b_{x}(t)<1$ for almost all $t \in [\bar{t}-\epsilon,\bar{t}+\epsilon]$ for some $\epsilon > 0$ and some $\bar{t}<t_{max}$. Then $I_{x}(t)<I_{1}(t)$ for all $t \in [\bar{t},t_{max}]$. 
\end{lemma}
\begin{proof}
Let the time of divergence of the interventions $b_1$ and $b_x$ be given by $t_{d}=\inf \{t \in (t_{i},t_\mathrm{max})\mid b_{x}(t)<1\}$. Note that this infimum must exist by completeness, because it is taken over a bounded and non-empty set. It is easy to see that $b_{x}(t)=1$ and that $I_{x}(t)=I_{1}(t)$ for all $t<t_{d}$. Also, by the right continuity of $b_{x}$, there exists an interval $(t_{d},t_{d}+\epsilon)$ such that $I_{1}>I_{x}$ and $b_1\beta S_{1}I_{1}>b_{x}\beta S_{x}I_{x}$. Suppose there exists a minimal $t_e\in (t_d,t_{\max})$ such that $I_{x}(t_e)=I_{1}(t_e)$. We will find a contradiction. Note that $I_x<I_1$ in $(t_d,t_e)$. Because $I_1$ is monotonic and continuous in $(-\infty,t_{\max})$, we can invert $I_1$ and define $\hat{t}=I_{1}^{-1}(I_{x}(t))$ for all $t \in (-\infty,t_{\max})$.

We can see that $\frac{dI_1}{dt}(\hat{t})$ is the growth rate of the infectious class under the null intervention when $I_1=I_x(t)$. It follows that $\frac{dI_1}{dt}(\hat{t})<\frac{dI_x}{dt}(t)$ for some value of $t$, otherwise $I_1$ would always dominate $I_x$, and $t_e$ would not exist. Let $t_{\inf}=\inf\{t\in (t_d,t_e)\mid \frac{dI_1}{dt}(\hat{t})<\frac{dI_x}{dt}(t)\}$. This implies that $S_{1}(\hat{t})\leq S_{x}(t)$ for some interval including $t_{\inf}$ that can be maximally extended to the left as $(t_{\min},t_{\inf}]$. By continuity of $S$, $S_1(\hat{t}_{\min})=S_x(t_{\min})$.

If $t_{\min}\neq t_d$, then $t_{\min}> \hat{t}_{\min}$, and $R_1(\hat{t}_{\min})=R_x(t_{\min})$, but this is impossible because by construction $\frac{dI_1}{dt}(\hat{t})>\frac{dI_x}{dt}(t)$ for all $t\in (t_d,t_{\min})$, and therefore $R_1(\hat{t}_{\min})<R_x(t_{\min})$.

If $t_{\min}= t_d$, then $t_{\min}= \hat{t}_{\min}$. Also, $R_1(\hat{t})>R_x(t)$ for all $t \in (t_d,t_\mathrm{inf})$. But that is impossible because $\frac{dR_{1}(\hat{t}_{d})}{dt}=\frac{b(t_d)\beta SI-\gamma I}{\beta SI-\gamma I}\gamma I<\gamma I= \frac{dR_{x}(t_{d})}{dt}$, and trivially $R_1(\hat{t}_d)=R_x(t_d)$.
\end{proof}
\begin{lemma}[Wait, maintain, suppress]
\label{lemma:wait-maintain-suppress}
Let $b_x$ be any intervention, and $I^{max}_x=\max\{ I_x(t)\}$. There exists an intervention $b_y$ with same beginning and duration as $b_x$ of the form

\begin{equation}
\begin{aligned}
\label{eqn:up-hold-down}
b_y(t)=  \begin{cases}
	1 & t\in[t_{i},t_{i}+g\tau) \\
    \frac{\gamma}{\beta S}, & t\in[t_{i}+g\tau,t_{i}+f\tau) \\
	0, & t\in[t_{i}+f\tau,t_{i}+\tau] \\
  \end{cases}
\end{aligned}
\end{equation}

For some $g<f$ and $g,f\in(0,1]$, such that $I^{max}_y \leq I^{max}_x$.
\end{lemma}

\begin{proof}
Define $I^{\tau}_x=\max\{ I_x(t) \mid t\in [t_i,t_f]\}$. Following Lemma \ref{lemma:do_nothing}, take $g$ such that $t_{i}+g\tau=I_{1}^{-1}(I_{x}^\tau)$. Now take $f$ such that $(1-f)\tau=\frac{1}{\gamma}\log{\frac{I_{x}^{\tau}}{I_x(t_f)}}$. From Lemma \ref{lemma:do_nothing} it is clear that $I_y(t)\geq I_x(t)$ for $t\in[t_{i}\tau,t_{i}+g\tau)$. By construction $I_y(t)=I_{x}^{\tau}$ for $t\in[t_{i}+g\tau,t_{i}+f\tau)$, and therefore the inequality remains true. Finally, for  $t\in[t_{i}+f\tau,t_{i}+\tau]$, because $I_y(t)$ decays faster than $I_x(t)$, then if $I_y(t)< I_x(t)$ at any time, then $I_y(t_f)< I_x(t_f)$, but by construction $I_y(t_f)= I_x(t_f)$. By Lemma \ref{lemma:more_infected_less_susceptible}, $S_y(t_f)< S_x(t_f)$, and by Lemma \ref{lemma:fuel-fire}, $\Imax_y<\Imax_x$.

It is possible however that this proposed $b_y$ is not a viable intervention if $S_y(t^\mathrm{crit}_y)=\Scrit$ for some $t^\mathrm{crit}_y \in [t_{i}+g\tau,t_{i}+f\tau)$, as this would require $b_y$ to assume values larger than $1$. If that is the case we can define $b_{\bar{y}}$, by taking $f$ such that $t_i+f\tau=t^\mathrm{crit}_y$. Note that because by the end of the intervention $b_{\bar{y}}$ the infectious class will monotonically decrease because $S_{\bar{y}}\leq \Scrit$, this implies that $\Imax_{\bar{y}}= I_{\bar{y}}^{\tau}= I_{x}^{\tau} \leq \Imax_x$.
\end{proof}

\begin{theorem}[Maintain, then suppress]
\label{theorem:maintain-then-suppress}
An intervention $\bopt(t)$ such that $I_{opt}^{max}\leq I_{x}^{max}$ for any intervention $b_{x}(t)$ must take the form

\begin{equation}
\begin{aligned}
\label{eqn:optimal-supp}
\bopt(t)=  \begin{cases}
    \frac{\gamma}{\beta S}, & t\in[t_{i},t_{i}+f\tau) \\
	0, & t\in[t_{i}+f\tau,t_{i}+\tau] \\
  \end{cases}
\end{aligned}
\end{equation}
for some value of $t_{i}\in  \mathbb{R}$ and $f \in (0,1]$.
\end{theorem}
\begin{proof}
If $b_x(t)$ is an optimal intervention, then Lemma \ref{lemma:wait-maintain-suppress} ensures that $b_x(t)$ is of the form given by equation \ref{eqn:up-hold-down}. Now consider the strategy $\bopt(t)$ of the form given by equation \ref{eqn:optimal-supp} (main text equation \ref{eqn:optimal}) with $t_i^\mathrm{opt}=t_i^{x}+g\tau$ and $f^\mathrm{opt}=f^{x}-g$. This new intervention functions exactly like $b_x$ for the entire duration of $b_x$, but then $\bopt$ is held at $0$ for a little longer further reducing $I_\mathrm{opt}$. This means that $S_\mathrm{opt}(t_{f}^\mathrm{opt})=S_{x}(t_{f}^x)$ and $I_\mathrm{opt}(t_{f}^\mathrm{opt})\leq I_{x}(t_{f}^x)$ and therefore by Lemma \ref{lemma:fire-fire}, $\Imax_\mathrm{opt}\leq \Imax_x$.
\end{proof}

\begin{remark}[The strategy is set at the start] Because during an intervention \bopt{} the susceptible class is depleted at a constant rate  for all $t \in [t_{i},t_{i}+f\tau)$ , \bopt{} can be written as

\begin{equation}
\begin{aligned}
\label{eqn:bopt_time}
\bopt(t)=  \begin{cases}
    \frac{\gamma}{\beta \big(S_{t_i} - I_{t_i} \gamma (t-t_i)\big)}, & t\in[t_{i},t_{i}+f\tau) \\
	0, & t\in[t_{i}+f\tau,t_{i}+\tau] \\
  \end{cases}
\end{aligned}
\end{equation}
For $S_{t_i}=S(t_i)$ and $I_{t_i}=I(t_i)$.
\end{remark}

\subsection{Results of an optimal intervention}

It follows that, given an optimal intervention $b$ of duration $\tau$ and depletion fraction $f$ begun at time $t_i$ and ending at $t_f = t_i + \tau$:

\begin{equation}
    S(t_f) = S(t_i) - \gamma \tau f I(t_i)
\label{eqn:opt-s-final}
\end{equation}

\begin{equation}
    I(t_f) = I(t_i) \exp\left[-\gamma \tau (1-f)\right]
\label{eqn:opt-i-final}
\end{equation}

In an SIR system without intervention that begins with a wholly susceptible population, we have the relation:

\begin{equation}
\label{eqn:S_I_relation}
    I(S) = 1 - S + \frac{1}{\Ro} \log(S)
\end{equation}

And so:
\begin{equation}
    I(t_i) = 1 - S(t_i) + \frac{1}{\Ro} \log(S(t_i))
\end{equation}

And we can likewise write $S(t_f)$ and $I(t_f)$ in terms of $S(t_i)$:

\begin{equation}
    S(t_f) = S(t_i) - \gamma \tau f \Big(1 - S(t_i) + \frac{1}{\Ro} \log(S(t_i))\Big)
\end{equation}

\begin{equation}
    I(t_f) = \Big(1 - S(t_i) + \frac{1}{\Ro} \log(S(t_i))\Big) \exp\left[-\gamma \tau (1-f)\right]
\end{equation}

This in turn allows us to express $\Imax(t_f)$ purely in terms of $S(t_i)$ and the parameters, though the expression is long:

\begin{equation}
\label{eqn:imax_in_Si}
    \begin{aligned}
        \Imax(t_f) &= \Big(1 - S(t_i) + \frac{1}{\Ro} \log(S(t_i))\Big) \exp\left[-\gamma \tau (1-f)\right] \\ 
        &+ S(t_i) - \gamma \tau f \Big(1 - S(t_i) + \frac{1}{\Ro} \log(S(t_i))\Big) \\
        &- \frac{1}{\Ro} \log\Big(S(t_i) - \gamma \tau f \Big(1 - S(t_i) + \frac{1}{\Ro} \log(S(t_i))\Big)\Big) \\
        &- \frac{1}{\Ro} + \frac{1}{\Ro} \log\Big(\frac{1}{\Ro}\Big)
    \end{aligned}
\end{equation}

\begin{remark}[Continuity of $I(t_i), \Imax(t_f)$]
\label{rem:continuity-of-I-max-candidates}
Notice that both $I(t_i)$ and $\Imax(t_f)$ are continuous functions of $S(t_i)$, since they are both linear combinations of continuous functions of $S(t_i)$ and since $S(t_i)$ is continuous in $t_i$, they are continuous functions of $t_i$.
\end{remark}

\begin{lemma}[Don't be late]
\label{lemma:dont-be-late}
Let $t_{p}$ be the infimum of times such that $I(t) = \Imax$, and let $t_i$ be the start of an optimal intervention. Then $t_{p} \ge t_i$. That is, $\Imax{}$ cannot occur before the intervention begins for an optimal intervention.
\end{lemma}

\begin{proof}
Since $b(t) = 1$ for $t<t_i$ $t_{p} < t_i$ necessarily implies that $S(t_i) < \Scrit$, that is, the epidemic is already declining when the intervention begins and will never grow again, regardless of the intervention approach (since $b(t) \le 1$, we cannot force a declining epidemic to grow). That in turn implies $\Imax = I^\text{peak}$, which we can with certainty improve upon. So $\Imax$ must occur during or after the intervention.
\end{proof}

\begin{corollary}[Start with fuel]
\label{corr:start-with-fuel}
It is immediate that $S(t_i) > \Scrit$.
\end{corollary}

\begin{corollary}[Peak early]
\label{corr:peak-early}
Since $I(t) \le I(t_i)$ for $t \in [t_i, t_i + \tau]$ during an optimal intervention, if $\Imax$ occurs during the intervention, $\Imax = I(t_i)$
\end{corollary}

\begin{theorem}[Twin Peaks] 
\label{theorem:twin-peaks}
Let $b_{x}(t)$ be an optimal intervention. Then $I_{x}(t)\le \Imax_x$ for all $t \in (-\infty,t_{f}]$, with equality for $t \in [t_i, t_i + \tau f]$, and furthermore $I_x(t_p) = \Imax_x$ for some $t_p \in [t_f, \infty)$ with $t_p = t_f$ only if $f = 1$. 

That is, if $0 < f < 1$, there will be a plateau during the intervention followed by a peak of equal height that occurs strictly after the intervention finishes. If $f = 0$, there will be two peaks of equal height, one at the start of the intervention and one strictly after it finishes, and if $f = 1$ there will be a plateau during the intervention with no subsequent peak.
\end{theorem}

\begin{proof}
Let $b_{x}$ be an optimal intervention of the form given by equation \ref{eqn:optimal-supp}. By Lemma \ref{lemma:dont-be-late}, $\Imax_{x}$ must occur during or after the intervention. First let us assume that the $\Imax_{x}$ occurs during the intervention and is never again attained after the intervention. If $f=1$, this implies that the whole intervention is a plateau and no further peaks occur. For $f<1$ we can build a new intervention $b_{y}$ of the same form as $b_x$ but that starts at $t_i-\epsilon$ rather than at $t_i$. From the continuity of $\Imax$ in $t_i$ for any intervention, it follows that for some $\epsilon$ small enough, $\Imax_y(t_f)$ (the post-intervention maximum of $I(t)$) must still be smaller than $I^\tau_y = I(t_i - \epsilon)$ (the maximum during the intervention), but because $I(t_i-\epsilon)<I(t_i)$, Corollary \ref{corr:peak-early} implies that our new $b_y$ outperforms $b_x$, which contradicts the optimality of $b_x$.

Now let us assume that $\Imax_{x}$ occurs after the intervention and is larger than any value of $I_x$ during the intervention. If $f>0$, we can once again build a new intervention function $b_y$

\begin{equation}
\begin{aligned}
b_y(t)=  \begin{cases}
	1 & t\in[t_{i},t_{i}+\epsilon) \\
    \frac{\gamma}{\beta S}, & t\in[t_{i}+\epsilon,t_{i}+f_y\tau) \\
	0, & t\in[t_{i}+f_y\tau,t_{i}+\tau] \\
  \end{cases}
\end{aligned}
\end{equation}

With $f_y$ chosen such that $(1-f_y)\tau=\frac{1}{\gamma}\log{\frac{I_{x}(t_i+\epsilon)}{I_x(t_f)}}$. For sufficiently small $\epsilon$, $I_y(t_{i}+\epsilon)<\Imax_x$. Also, following Lemma \ref{lemma:wait-maintain-suppress}, $\Imax_y < \Imax_x$ and therefore $b_y$ outperforms $b_x$, which once again contradicts the optimality of $b_x$.

Finally, if $f = 0$, we build yet another $b_y$ of the same form as $b_x$ but that starts at $t_i+\epsilon$ rather than at $t_i$. Because intervention $b_y$ starts out with a smaller susceptible fraction than intervention $b_x$, and an increase in the infected fraction that is smaller than the susceptible fraction decrease, it follows from equation \ref{eqn:Imax-t} that $\Imax_y < \Imax_x$ and therefore $b_y$ outperforms $b_x$, which once again contradicts the optimality of $b_x$.
\end{proof}

\begin{corollary}[The longer the intervention, the earlier it starts]
\label{corr:longer_earlier}
If $t_i^{opt}(\tau)$ is the starting time of the optimal intervention with duration $\tau$, let $S_i^{opt}(\tau)=S(t_i^{opt}(\tau))$. Then $S_i^{opt}(\tau)$ is a non-decreasing function of $\tau$.
\end{corollary}
\begin{proof}
Trivially a longer optimal intervention must outperform a shorter optimal intervention. But by Theorem \ref{theorem:twin-peaks} $I_i^\mathrm{opt}=\Imax$ for any optimal intervention, which implies that $I_i^\mathrm{opt}(\tau)$ is a non-increasing function of $\tau$ and $S_i^\mathrm{opt}(\tau)$ is a non-decreasing function of $\tau$.
\end{proof}

\begin{corollary}[If you don't have much gunpowder, don't shoot until you see the whites of their eyes]
\label{corr:bunker-hill}
Given the optimal intervention of duration $\tau$, $\bopt^{\tau}$, denote the initial time of that intervention by $t_{i}^{\tau}$. As $\tau \to 0$, $t_{i}^{\tau} \to t_{crit}$.
\end{corollary}
\begin{proof}
From equation \ref{eqn:imax_in_Si} we can see that $\Imax_\mathrm{opt}$ is a continuous function of $\tau$, and that as $\tau \to 0$, the effect of intervention defined as $\Imax_{1}-\Imax_\mathrm{opt}\to 0$. We know from Theorem \ref{theorem:twin-peaks} that $I_\mathrm{opt}(t_i)=\Imax_\mathrm{opt}$, and therefore, as $\tau \to 0$, $I_\mathrm{opt}(t_i) \to I^{max}_{1}$, which implies $t_{i}^{\tau} \to t_\mathrm{crit}$.
\end{proof}

\begin{corollary}[No need to burn all the fuel] 
\label{corr:f1}
After an optimal intervention, $S(t_f) \ge \Scrit$, with equality if and only if $f = 1$.
\end{corollary}
\begin{proof}
By Theorem \ref{theorem:twin-peaks}, $I(t_i) = \Imax(t_f) = \Imax$. If $f < 1$, we must have $I(t_f) < I(t_i) = \Imax$, so in order for the epidemic to reach $\Imax(t_f) = \Imax$, we must have $S(t_f) > \Scrit$. If $f = 1$, then $I(t_f) = I(t_i) = \Imax$, so then must have $S(t_f) = \Scrit$, otherwise the epidemic would grow to a peak above $I(t_i) = I(t_f)$, which would be a contradiction of Theorem \ref{theorem:twin-peaks}.
\end{proof}

\begin{corollary}[Putting out existing fire can only do so much]
\label{corr:full_supression_optima}
Consider a full suppression intervention of duration $\tau$ defined by $b_0(t)=0$ for all $t \in [t_i,t_i+\tau]$. For every $\tau$ there is a $t_i$ that minimizes $\Imax_0$. Consider these optimized full suppression interventions. Then, as $\tau \to \infty$, the maximum infectious prevalence $\Imax_0\to \frac{1}{2}+\frac{1}{2\Ro}\Big(\log \left(\frac{1}{\Ro}\right)-1\Big)$. In other words, full suppression interventions have a limit in how much they can reduce $\Imax$.
\end{corollary}

\begin{proof}
From the proof of Theorem \ref{theorem:twin-peaks} for $f=0$, it follows that full suppression interventions have an optimal start time $t_i$, and that $I_0(t_i)=\Imax_0$. Also, because no new infections occur during a full suppression intervention, $S_0(t_i)=S_0(t_f)$ so substituting into equation \ref{eqn:Imax-t}

\begin{equation}
\Imax_0(t_i) = 
I_0(t_f) + S_0(t_i) - \frac{1}{\Ro} \log\Big(S_0(t_i)\Big) - \frac{1}{\Ro} + \frac{1}{\Ro} \log\Big(\frac{1}{\Ro}\Big)
\end{equation}

We the further use equation \ref{eqn:S_I_relation} to substitute $S_0(t_i)$ and take $\tau \to \infty$ such that $I_0(t_f)\to 0$ which finally yields

\begin{equation}
\Imax_0= 
\frac{1}{2}+\frac{1}{2\Ro}\Big(\log \left(\frac{1}{\Ro}\right) - 1\Big)
\end{equation}

\end{proof}

\begin{corollary}[But when in doubt, put out the fire]
\label{corr:but-just-put-it-out}
Let $b_0$ be the optimized full suppression intervention with duration $\tau$ and starting time $t_i^0$. For a full suppression intervention $b_{\hat{0}}$ that also has duration $\tau$ but that has a starting time $t_i^{\hat{0}} \in [t_i^0,t_{crit}]$, the infected fraction peak is $\Imax_{\hat{0}}=I_{\hat{0}}(t_i^{\hat{0}})$, moreover no intervention starting at that same time can attain a lower peak.
\end{corollary}

\begin{proof}
It suffices to show that delaying a full suppression intervention by an infinitesimal $\epsilon$ diminishes the secondary peak. From equation \ref{eqn:partials} with $f=0$, it is clear that

\begin{equation}
\begin{aligned}
\pdv{}{t_i} \Imax(t_f) &= - \frac{S'(t_i)}{\Ro S_i} + e^{-\gamma \tau}I'(t_i) + S'(t_i)\\
\pdv{}{t_i} \Imax(t_f) &= \gamma I_i + e^{-\gamma \tau}(\beta I_i S_i-\gamma I_i)- \beta I_i S_i\\
\pdv{}{t_i} \Imax(t_f) &= (e^{-\gamma \tau}-1)(I'(t_i))
\end{aligned}
\end{equation}
But $I'(t_i)>0$ for $t_i<t_\mathrm{crit}$ and $ (e^{-\gamma \tau}-1)<0$ which means that delaying the full suppression decreases the post-intervention peak. Trivially, no intervention can attain a global \Imax{} value lower than its initial prevalence $I(t_i)$, which concludes the proof.
\end{proof}

\begin{theorem}[Twin peaks for fixed control interventions]
\label{theorem:twin-peaks-fixed}
Let $b_x$ be an optimized fixed control intervention of duration $\tau$ that starts at some time $t_i$ and has strictness $\sigma$. That is, $b_x$ outperforms any other fixed control intervention of duration $\tau$ regardless of their $t_i$ or $\sigma$. Then $I_x=\Imax_x$ at exactly two time points, one during the intervention ($t\in [t_i,t_i+\tau]$) and one after the intervention ($t\in (t_i+\tau,\infty)$).
\end{theorem}
\begin{proof}
Because during the intervention and after the intervention, the time course is that of an SIR (albeit with a modified $\Ro$ during the intervention), it is clear that there can be at most one local $I$ maximum during the intervention and at most one after the intervention.

First let us assume that the infectious peak during the intervention is higher than any peak after the intervention. By continuity it is possible to construct a new fixed control intervention that starts a little earlier and is slightly stricter such that the peak during the intervention is reduced but the peak after the intervention is still below the intervention peak. This new intervention outperforms $b_x$ which contradicts the optimality of $b_x$.

Now let us assume that the infectious peak during the intervention is lower than any peak after the intervention. We will show that a small change in the $t_i$ of the intervention will lower the post-intervention peak. It suffices to show that $\pdv{}{t_i}\Imax\neq 0$. For a fixed control intervention we have

\begin{equation}
\begin{aligned}
\pdv{}{t_i}\Imax &= \frac{(1-\sigma)}{\sigma \Ro S_x(t_f)}*\pdv{}{t_i}S_x(t_f)
\end{aligned}
\end{equation}

Which when set to $0$ implies $\pdv{}{t_i}S_x(t_f)=0$. This in turn can be substituted into the equation that relates $I$ to $S$ under a fixed control intervention

\begin{equation}
\begin{aligned}
I(t_f) &= -S(t_f)+\frac{1}{\sigma \Ro}\log (S(t_f))+I(t_i)+S(t_i)-\frac{1}{\Ro}\log (S(t_i))
\end{aligned}
\end{equation}

Which leads to $\pdv{}{t_i}I_x(t_f)=0$ and finally $\pdv{}{t_i}R_x(t_f)=0$ but this is impossible because the fraction of recovered individuals by the end of a fixed control intervention must decrease as the intervention starts earlier. This completes the proof.

\end{proof}

\subsection[First-order conditions for the optimal intervention]{First-order conditions for $t_i$, $f$}
Applying equations \ref{eqn:opt-s-final}, \ref{eqn:opt-i-final}, the partial derivatives of $\Imax(t_f)$ with respect to $f$ and $t_i$ are given by:
\begin{equation}
\begin{aligned}
\label{eqn:partials}
\pdv{}{f} \Imax(t_f) &= \frac{(\gamma  \tau I_i)}{\Ro\Big(S_i - \gamma \tau f I_i\Big)} + \gamma \tau I_i e^{-(1-f)\gamma \tau} - \gamma \tau I_i \\ \\ 
\pdv{}{t_i} \Imax(t_f) &= - \frac{S'(t_i) - f \gamma \tau I'(t_i)}{\Ro \Big(S_i - f \gamma \tau I_i \Big)} + e^{-(1-f)\gamma \tau}I'(t_i) - f \gamma \tau I'(t_i) + S'(t_i)
\end{aligned}
\end{equation}

Substituting the values of $\dv{S}{t}, \dv{I}{t}$:
\begin{equation}
\begin{aligned}
    \pdv{}{t_i} \Imax(t_f) &= -\frac{\left(-\beta S_i I_i \right)-f\gamma\tau\left(\beta S_i I_i - \gamma I_i \right)}{\Ro\left(S_i -f\gamma\tau I_i \right)} \\ 
    & + e^{-\left(1-f \right)\gamma\tau}\left(\beta S_i I_i- \gamma I_i \right) \\ 
    &- f\gamma\tau\left(\beta S_i I_i - \gamma I_i \right) -\beta S_i I_i
\end{aligned}
\end{equation}

Setting equal to zero and simplifying: 

\begin{equation}
\begin{aligned}
\label{eqn:partial_f}
\pdv{}{f} \Imax(t_f) &= \frac{1}{\Ro\Big(S_i - f\gamma \tau I_i\Big)} + \gamma \tau e^{-(1-f)\gamma \tau} - 1 = 0 \\
\end{aligned}
\end{equation}

\begin{equation}
\begin{aligned}
\label{eqn:partial_t}
\pdv{}{t_i} \Imax(t_f) &= -\frac{\left(-\beta S_i \right)-f\gamma\tau\left(\beta S_i - \gamma \right)}{\Ro\left(S_i - f \gamma\tau I_i \right)} + e^{-\left(1-f \right)\gamma\tau}\left(\beta S_i - \gamma \right) \\ 
    &-f\gamma\tau\left(\beta S_i - \gamma \right) -\beta S_i = 0
\end{aligned}
\end{equation}

While these first-order conditions do not yield a simple closed form for the optimal $f$ and $t_i$, we use them in the proofs above and below to establish properties of the optimal strategy.

\begin{lemma}[If you have lots of gunpowder, shoot early]
\label{lemma:lots-of-powder-shoot-early}
For an optimal intervention $\bopt^\tau$ acting on a SIR with full susceptibility as its initial condition, as $\tau \to \infty$, $S_{opt}^{\tau}(t_{i})\to 1$; that is, we intervene almost immediately, when almost all the population is initially susceptible and a long intervention is possible.
\end{lemma}
\begin{proof}
Define an auxiliary intervention $b_{x}^{\tau}$ with $\tau > \frac{1-S_{crit}}{\gamma I^{max}}$ such that
\begin{equation}
\begin{aligned}
b_{x}^{\tau}(t)=\begin{cases}
    \frac{\gamma}{\beta S}, & \text{if } t>t_x \text{ and } S>S_{crit} \\
	1, & \text{elsewhere}
  \end{cases}
\end{aligned}
\end{equation}
with $t_x=I^{-1}_{1}(\frac{1-\Scrit}{\gamma \tau})$. It is clear that such an intervention has a duration of $\tau$ or less, and that by the end of such an intervention $S_x\leq \Scrit$. Therefore $\Imax_{x}=I_{x}(t_x)=\frac{1-\Scrit}{\gamma \tau}$, but by definition $\Imax_\mathrm{opt}\leq \Imax_{x}$. Combining with Corollary \ref{corr:peak-early}, $I_\mathrm{opt}^{\tau}(t_i)\leq \frac{1-\Scrit}{\gamma \tau}$ so as $\tau \to \infty$ both $I_\mathrm{opt}^{\tau}(t_{i})\to 0$ and $S_\mathrm{opt}^{\tau}(t_{i})\to 1$.
\end{proof}
\begin{theorem}[Sometimes maintain, always suppress]\label{theorem:always-suppress}
For an SIR with full susceptibility as its initial condition, it is the case that any optimal intervention with positive duration as defined by equation \ref{eqn:optimal-supp} has $f < 1$. In other words, the optimal intervention for an emerging pathogen ($S(0) \to 1$, $I(0)\to 0$) always has a total suppression phase. Also, for any $\Ro \in (1,\infty)$ there is a $\tau_{crit}$ such that the optimal intervention is always a full suppression intervention ($f=0$) if $\tau<\tau_{crit}$. Moreover, $\lim_{\Ro \to 1}\tau_{crit}=\infty$ and  $\lim_{\Ro \to \infty}\tau_{crit}=0$.
\end{theorem}

\begin{proof}
Suppose $f=1$. From equation \ref{eqn:partial_f}, 
\begin{equation}
\begin{aligned}
\frac{1}{\Ro\Big(S_i - \gamma \tau I_i\Big)} + \gamma \tau  - 1 = 0 
\end{aligned}
\end{equation}
But following Corollary \ref{corr:f1}, $S_i - \gamma \tau I_i=\Scrit$, which substituting in the previous equation implies that $\gamma \tau=0$, which is impossible. Therefore $f$ cannot be $1$.

Now let us investigate if $f=0$ can be a local minimum. Once again, from the first order condition expressed in equation \ref{eqn:partial_f} we require that

\begin{equation}
\begin{aligned}
\label{eqn:Si_f0}
S_i &=\frac{1}{\Ro x}\\
x&=1-e^{-\gamma \tau}\\
\end{aligned}
\end{equation}

The second order condition for an extremum to be a minimum is given by

\begin{equation}
\begin{aligned}
\pdv{^2}{f^2}\Imax &=I_i (\gamma \tau)^2 e^{- \gamma \tau (1-f)}+\frac{(\gamma \tau I_i)^2}{\Ro (S_i-\gamma \tau f I_i)^2}>0
\end{aligned}
\end{equation}

Note that the condition is always satisfied, and therefore any extremum is always a minimum. This also ensures that any minimum must be unique and that $\Imax$ grows monotonically as $f$ moves away from $f^\mathrm{opt}$.

From Theorem \ref{theorem:twin-peaks} we have $I_i=\Imax$. Substituting this and equations \ref{eqn:Si_f0} and \ref{eqn:S_I_relation} into equation \ref{eqn:imax_in_Si}, we get an implicit function for $x$ (and therefore for $\tau$) such that $f=0$ is a minimum:

\begin{equation}
\begin{aligned}
\label{eqn:implicit_tau}
0 &=x(\log (x)+x\log (\Ro x)-\Ro x)+1
\end{aligned}
\end{equation}

We now show that $\Ro \in (1,\infty)$ implies $x\in (0,1)$ which in turn implies a $\taucrit$ that satisfies the relation.

From equation \ref{eqn:implicit_tau} it is easy to see that as $\Ro \to 1$, then $x \to 1$, and as $\Ro \to \infty$, then $x \to 0$. So showing that $x(\Ro)$ is a continuous, decreasing function of $\Ro$ on $\Ro \in (1, \infty)$ will establish the existence of $\taucrit$.

By implicit differentiation we have that

\begin{equation}
\begin{aligned}
\label{eqn:implicit_tau_diff}
\pdv{x}{\Ro}=\frac{x^3 \big( \frac{\Ro-1}{\Ro} \big)}{x+x^2-x \log (x) -2}
\end{aligned}
\end{equation}

For positive $x$ the numerator is positive and the denominator is always negative if $x\in (0,1)$. Therefore, for $\Ro \in (1,\infty)$ there exists $x\in (0,1)$ that satisfies equation \ref{eqn:implicit_tau}.

To complete the proof, it suffices to show that $f^\mathrm{opt}$ is an increasing function of $\tau$ when $f=0$, as this implies that $f^\mathrm{opt}(\tau_{crit})=0$ and cannot go above $0$ for $\tau<\tau_\mathrm{crit}$. We start by letting  $f^\mathrm{opt}$, $S_i^\mathrm{opt}$, and $I_i^\mathrm{opt}$ be functions of $\tau$, applying equation \ref{eqn:partial_f}, implicitly differentiating with respect to $\tau$, and setting $f=0$. After a bit of algebra we obtain
\begin{equation}
\begin{aligned}
\big[ \tau e^{-\gamma \tau}+ \frac{\tau I_i^\mathrm{opt}}{\Ro (S_i^\mathrm{opt})^2}\big] *\pdv{}{\tau}f^\mathrm{opt}=e^{-\gamma \tau}+ \frac{1}{\Ro \gamma (S_i^\mathrm{opt})^2}*\pdv{}{\tau}S_i^\mathrm{opt}
\end{aligned}
\end{equation}

Note that $\pdv{}{\tau}f^\mathrm{opt}$ is being multiplied by a strictly positive quantity on the left-hand side of the equation. On the right-hand side of the equation there is a positive quantity ($e^{-\gamma \tau}$) added to $\pdv{}{\tau}S_i^\mathrm{opt}$ multiplied by another positive quantity ($\frac{1}{\Ro (S_i^\mathrm{opt})^2}$), but by Corollary \ref{corr:longer_earlier}, $\pdv{}{\tau}S_i^\mathrm{opt}>0$. It follows that $\pdv{}{\tau}f^\mathrm{opt}>0$ when $f^\mathrm{opt}=0$ which completes the proof.

\end{proof}

\subsection{A general classification of interventions}
From equation \ref{eqn:Imax-t} it is clear that there are two fundamental methods of reducing the peak of an epidemic: depleting the infected fraction and depleting the susceptible fraction. We have shown that different interventions achieve peak reduction with different combinations of those methods. Our optimal intervention, for example, is characterized by a pure susceptible depletion phase followed by a pure infected depletion phase. Full suppression interventions, in contrast, operate solely by depleting the infected fraction. We observe that interventions can be classified in terms of how much they rely on depleting the susceptible fraction versus depleting the infected fraction.

The effect of an intervention can be understood as the infectious peak if no intervention were to take place minus the infectious peak given the intervention. In a more formal notation, an intervention $b_x$ has an effect $\Imax-\Imax_x=\Imax(I_x(t_i),S_x(t_i))-\Imax(I_x(t_f),S_x(t_f))=\Delta_x(t_f)$. By applying equation \ref{eqn:Imax-t} we obtain

\begin{equation}
\begin{aligned}
\Delta_x(t_f)&=  -\Big[I_x(t_f)-I_x(t_i)\Big]-\Big[G\Big(S_x(t_f)\Big)-G\Big(S_x(t_i)\Big)\Big]\\
G(S) &=S-\frac{1}{\Ro}\log(S)\\
\end{aligned}
\end{equation}

Then if $-\Big[I_x(t_f)-I_x(t_i)\Big]>-\Big[G\Big(S_x(t_f)\Big)-G\Big(S_x(t_i)\Big)\Big]$ we can say that the intervention $b_x$ overall relies more on infected depletion, whereas if the opposite is true we can say that it relies more on susceptible depletion. Moreover, by applying the fundamental theorem of calculus, we obtain

\begin{equation}
\begin{aligned}
\int_{t_i}^{t_f}\Delta_{x}'(t)dt&= -\int_{t_i}^{t_f}I'_x(t)+S'_x(t)\Big(1-\frac{1}{\Ro S_x(t)}\Big)dt 
\end{aligned}
\end{equation}
Which allows us to look at a certain time $t \in [t_i,t_f]$ and say that if
\begin{equation}
\begin{aligned}
I'_x(t)<S'_x(t)\Big(1-\frac{1}{\Ro S_x(t)}\Big)
\end{aligned}
\end{equation}
Then at that moment $t$, the intervention $b_x$ acts more by depleting the infected fraction than by depleting the susceptible fraction. The condition can be simplified to
\begin{equation}
\begin{aligned}
b_x(t)\Big(2\Ro S_x(t)-1 \Big)<1
\end{aligned}
\end{equation}
\printbibliography[
title = {Supplementary References},
section = 2]
\end{refsection}

\end{document}